%% file: proper.tex
\documentclass[11pt]{article}
\usepackage{amsmath,amssymb,amsfonts,amsthm,epsfig}
\usepackage[usenames,dvipsnames]{xcolor}
\usepackage{bm,xspace}
\usepackage{tcolorbox}
\usepackage{cancel}
\usepackage{fullpage}
\usepackage{liyang}
\usepackage{framed}
\usepackage{verbatim}
\usepackage{enumitem}
\usepackage{array}
\usepackage{multirow}
\usepackage{afterpage}
\usepackage{mathrsfs}
\usepackage{pifont} 
\usepackage{chngpage}
\usepackage[normalem]{ulem}
\usepackage{boxedminipage}
\usepackage{caption}
\usepackage{forest}

\def\colorful{0}

\ifnum\colorful=1
\newcommand{\violet}[1]{{\color{violet}{#1}}}

\fi
\ifnum\colorful=0
\newcommand{\violet}[1]{{{#1}}}

\fi

\newcommand{\BuildDT}{\textsc{BuildDT}}

\newcommand{\pruned}{\mathrm{pruned}}
\newcommand{\bias}{\mathrm{bias}}
\newcommand{\depth}{\mathrm{depth}}
\newcommand{\Prune}{\textsc{Prune}}

\makeatletter
\newtheorem*{rep@theorem}{\rep@title}
\newcommand{\newreptheorem}[2]{
\newenvironment{rep#1}[1]{
 \def\rep@title{#2 \ref{##1}}
 \begin{rep@theorem}\itshape}
 {\end{rep@theorem}}}
\makeatother

\newreptheorem{theorem}{Theorem}

\newcommand{\pparagraph}[1]{\bigskip \noindent {\bf {#1}}}

\begin{document}

\title{Properly learning decision trees in almost polynomial time\footnote{A preliminary version of this paper appeared in the proceedings of the 62nd Annual IEEE Symposium on Foundations of Computer Science (FOCS 2021).} 
 \vspace{15pt}}

\author{Guy Blanc \vspace{8pt} \\ \hspace{-5pt}{\sl Stanford} \and \hspace{10pt} Jane Lange \vspace{8pt} \\
\hspace{4pt}  {\sl MIT}
\and Mingda Qiao \vspace{8pt}\\ \hspace{-8pt} {\sl Stanford} 
\and Li-Yang Tan \vspace{8pt} \\ \hspace{-8pt} {\sl Stanford}}

\date{\vspace{15pt}\small{\today}}

\maketitle

\input{intro} 

\bibliography{proper}{}
\bibliographystyle{alpha}

\end{document}

%% file: intro.tex

\begin{abstract} 
We give an $n^{O(\log\log n)}$-time membership query algorithm for properly and agnostically learning decision trees under the uniform distribution over $\bn$.  Even in the realizable setting, the previous fastest runtime was $n^{O(\log n)}$, a consequence of a classic algorithm of Ehrenfeucht and Haussler.  

 Our algorithm shares similarities with practical heuristics for learning decision trees, which we augment with additional ideas to circumvent known lower bounds against these heuristics.  To analyze our algorithm, we prove a new structural result for decision trees that strengthens a theorem of   O'Donnell, Saks, Schramm, and Servedio.  While the OSSS theorem says that every decision tree has an influential variable, we show how every decision tree can be ``pruned"  so that {\sl every} variable in the resulting tree is influential.
\end{abstract} 

\thispagestyle{empty}

\newpage
\setcounter{page}{1}

\section{Introduction} 

Decision trees are a simple and effective way to represent boolean functions $f : \bits^n \to \bits$. Their logical, flow-chart-like structure makes them easy to understand, and they are the canonical example of an interpretable model in machine learning.  They are also fast to evaluate: the complexity of evaluating a decision tree on an input scales with the depth of the tree, which is often much smaller than the dimension $n$ of $f$.  

The algorithmic problem of converting a function $f$ into a decision tree representation $T$ has therefore been extensively studied by a number of communities spanning  both theory and practice.  Naturally, we would like $T$ to be as small as possible, ideally close to the optimal decision tree size of~$f$.  If we require $T$ to compute $f$ {\sl exactly}, this is unfortunately likely an intractable problem, even if  $T$ is allowed to be larger than the optimal decision tree for $f$: finding an approximately minimal decision tree for a given function is NP-hard~\cite{LR76,ZB00,Sie08,AH12}.

We therefore allow $T$ to err on a small fraction of inputs.  Our main result is a new algorithm for this problem:

\begin{theorem}
\label{thm:main} There is an algorithm which, given as input $\eps > 0$, $s\in \N$, and query access to a function $f : \bits^n\to\bits$ that is promised to be $\opt_s$-close to a size-$s$ decision tree, runs in time
\[ \tilde{O}(n^2) \cdot (s/\eps)^{O(\log((\log s)/\eps))} \]
and outputs a size-$s$ decision tree $T$ that w.h.p.~satisfies $\ds\Prx_{\mathrm{uniform}~ \bx}[T(\bx)\ne f(\bx)] \le \opt_s+\eps$. 
\end{theorem}

For $s = \poly(n)$ and $\eps \ge 1/\polylog(n)$, our algorithm runs in almost polynomial time, $n^{O(\log\log n)}$.  Even in the realizable setting  ($\opt_s = 0$), the previous fastest algorithms took quasipolynomial time, $n^{\Omega(\log n)}$, even for constant $\eps$.  This was  the state of the art even for algorithms with access to an explicit representation of $f$, rather than just query access. 

Another interesting setting is when the algorithm is only given uniform random examples labeled by $f$ rather than query access. For this setting, we have the following result:
\begin{theorem}
    \label{thm:monotone}
    In the context of \Cref{thm:main}, if $f$ is \emph{monotone}, our algorithm uses only random labeled examples $(\bx,f(\bx))$ where $\bx\sim \bn$ is uniformly random.
\end{theorem}


\subsection{Background and context} 

In the language of learning theory,~\Cref{thm:main} gives a query algorithm for properly and agnostically learning decision trees under the uniform distribution. We now overview previous algorithms for this and related problems. 

Ehrenfeucht and Haussler~\cite{EH89}, in an early paper following the introduction of the PAC learning model, gave an $n^{O(\log s)}$ time algorithm for properly learning size-$s$ decision trees. \cite{EH89}'s algorithm works in the more general distribution-free setting and only uses random examples.  On the other hand,~\cite{EH89} assumes the realizable setting, and their algorithm is not known to  extend to the agnostic setting.  This limitation is likely inherent: being an Occam algorithm, its analysis crucially relies on noiseless examples.   Furthermore,~\cite{EH89}'s algorithm is {\sl weakly} proper, in the sense that its decision tree hypothesis can be as large as $n^{\Omega(\log s)}$.  A (strongly) proper algorithm returns a hypothesis that belongs to the target concept class; in this case, a size-$s$ decision tree hypothesis for a size-$s$ decision tree target. 

Since the work of Ehrenfeucht and Haussler,  a couple of alternative algorithms for properly learning decision trees have been developed in the uniform-distribution setting.  These algorithms are quite different from~\cite{EH89}'s and from each  other.  Mehta and Raghavan~\cite{MR02} gave an $n^{O(\log s)}$ time algorithm that uses random examples, and more recently~\cite{BLT-ITCS} gave a $\poly(n) \cdot s^{O(\log s)}$ time membership query algorithm.  For the standard setting where $s = \poly(n)$, these runtimes are still $n^{\Omega(\log n)}$, just like~\cite{EH89}'s. 

Therefore, while~\cite{EH89}'s $n^{O(\log n)}$ runtime for properly learning polynomial-size decision trees has been matched twice in the uniform-distribution setting, it has remained unsurpassed for over three decades.  Furthermore, the analyses of all three algorithms are known to be tight: for each of them, there are targets for which the algorithm can be shown to require  $n^{\tilde{\Omega}(\log n)}$ time.  

\Cref{table} summarizes of how our algorithm compares with existing ones:  

\vspace{10pt} 
\begin{table}[H]
  \captionsetup{width=.9\linewidth}
\renewcommand{\arraystretch}{1.7}
\centering
\begin{tabular}{|c|c|c|c|c|}
\hline
  Reference  & Running time  & Hypothesis size   & ~~Access to target~~ &  Agnostic? \\ \hline
\cite{EH89} & $n^{O(\log s)}$ & $n^{O(\log s)}$ & Random examples & $\times$ \\ [.2em] \hline
 \cite{MR02} & $ n^{O(\log s)}$  & $s$ & Random examples & $\checkmark$ \\ [.2em]  \hline 
 \cite{BLT-ITCS} & $\poly(n)\cdot s^{O(\log s)}$ & $s^{O(\log s)}$ & Queries & $\times$ \\ [.2em] 
 \hline \hline
  {\bf This work} & ~~$\poly(n) \cdot s^{O(\log\log s)}$~~  & $s$ & Queries & $\checkmark$ \\ [.2em] \hline
\end{tabular}
\caption{Algorithms for properly learning size-$s$ decision trees.~\cite{EH89}'s algorithm works in the more general distribution-free setting, whereas all others, including ours, work in the uniform-distribution setting.}  
\label{table}
\end{table}

\paragraph{Improper algorithms.}  While the focus of our work is on proper learning, the problem of improperly learning decision trees, where the hypothesis is not required to itself be a decision tree, is also the subject of intensive study. 
Kusilevitz and Mansour~\cite{KM93} gave a polynomial-time membership query algorithm for learning polynomial-size decision trees under the uniform distribution; this was subsequently extended to the agnostic setting by Gopalan, Kalai, and Klivans~\cite{GKK08}.  Both works employ Fourier-analytic techniques, and their algorithms return the sign of a Fourier polynomial as their hypothesis.  

Other works on improper learning of decision trees include~\cite{Riv87,Blu92,Han93,Bsh93,BFJKMR94,HJLT96,JS06,OS07,KS06,KST09,HKY18,CM19}.

\paragraph{On the use of membership queries.} It would be preferable if our algorithm in~\Cref{thm:main} did not require membership queries and instead relied only on random examples.  However, there are well-known barriers to obtaining such an improvement of our algorithm, even an improper one and even just within the realizable setting. 

First, no such statistical query algorithm exists: any SQ algorithm for learning polynomial-size decision trees has to take $n^{\Omega(\log n)}$ time~\cite{BFJKMR94}.  Second, we observe that our $\poly(n)\cdot s^{O(\log\log s)}$ runtime is fixed-parameter tractable in `$s$'.  Obtaining a $\poly(n) \cdot \Phi(s)$ time algorithm that only uses random examples, for any growth function $\Phi$, would give the first polynomial-time algorithm for learning $\omega_n(1)$-juntas.  This would be a breakthrough on a notorious open problem~\cite{BL97}; current algorithms for learning $k$-juntas take time $n^{\Omega(k)}$~\cite{MOS04,Val15}.   





\section{Overview of our approach} 
\label{sec:overview} 
The starting point of our work is~\cite{BLT-ITCS}'s $\poly(n)\cdot s^{O(\log s)}$ time algorithm for the realizable setting.  We begin with a brief overview of their algorithm, followed by a description of how we obtain our improved $\poly(n) \cdot s^{O(\log\log s)}$ time algorithm in the realizable setting. We then explain how we extend our algorithm to the agnostic setting.

\paragraph{\cite{BLT-ITCS}'s greedy algorithm.}  At the heart of~\cite{BLT-ITCS}'s algorithm, as well as ours, is the notion of the {\sl influence} of a variable on a function.  For a function $f : \bits^n\to\bits$ and a variable $i\in [n]$, the influence of $i$ on $f$ is the quantity $\Pr[f(\bx)\ne f(\bx^{\sim i})]$, where $\bx\sim\bn$ is uniformly random and $\bx^{\sim i}$ denotes $\bx$ with its $i$-th coordinate rerandomized. 

\cite{BLT-ITCS} analyzes a simple greedy algorithm for constructing a decision tree $T$ for~$f$: 
\begin{enumerate}
    \item Using membership queries to $f$, identify the variable $i\in [n]$ with (approximately) the largest influence on $f$. 
    \item Query $x_i$ at the root of $T$.
    \item Build the left and right subtrees of $T$ by recursing on $f_{x_i=-1}$ and $f_{x_i=1}$ respectively. 
\end{enumerate}

\cite{BLT-ITCS} proved that growing this tree to size $s^{O(\log s)}$ yields a high-accuracy hypothesis for $f$. Their algorithm, like~\cite{EH89}'s, is weakly proper.

\paragraph{A near-matching lower bound.}   \cite{BLT-ITCS} provided a near-matching lower bound showing that their analysis of their algorithm is essentially tight.  They exhibited a size-$s$ decision tree target $f$ such that the tree grown by their algorithm has to reach size $s^{\tilde{\Omega}(\log s)}$ before achieving any nontrivial accuracy.  

\subsection{Our algorithm and its analysis} 

\cite{BLT-ITCS}'s algorithm formalizes the intuition, drawn from decision tree learning heuristics used in practice (e.g.~ID3, CART, C4.5), that the most influential variable is a ``somewhat good" root: the greedy strategy of recursively querying the most influential variable converges to a high-accuracy hypothesis at size $s^{O(\log s)}$.  Their lower bound establishes the limitations of this strategy.

At a high level, we obtain our improved algorithm by showing that {\sl there's an even better root among the $\polylog(s)$ most influential variables}.  Rather than committing to the single most influential variable as the root of our tree, we consider the set of $\polylog(s)$ most influential variables as candidate roots.  We prove the existence of a variable $x_i$ within this set such that growing a size-$s$ tree with $x_i$ as the root results in a high-accuracy hypothesis for $f$.

\subsubsection{Our key new tool: A pruning lemma for decision trees} 

The analysis of our algorithm is driven by a new structural lemma for decision trees.   This lemma generalizes a result of O'Donnell, Saks, Schramm, Servedio~\cite{OSSS05}---the {\sl OSSS inequality}---which is the crux of~\cite{BLT-ITCS}'s analysis of their algorithm: 

\begin{theorem}[OSSS inequality] 
\label{thm:OSSS}
Let $f : \bn\to\bits$ be a size-$s$ decision tree.  Then: 
\[ \max_{i\in [n]} \{\Inf_i(f)\} \ge \frac{\Var(f)}{2\log s}, \]
where $\Var(f)$ denotes the variance of the random variable $f(\bx)$. 
\end{theorem} 

In words, the OSSS inequality says that every small-size decision tree (that is not too biased) has an influential variable.  

Our new structural lemma shows that every decision tree can be ``pruned" so that {\sl every} variable in the resulting tree is influential.  Our notion of pruning is simple and is based on a single atomic procedure: one prunes a decision tree $T$ by iteratively replacing any of its internal nodes by one of the node's subtrees.

\begin{theorem}[Our pruning lemma for the realizable setting] 
    \label{thm:pruning intro}
Let $f$ be computable by a size-$s$ decision tree $T$ and $\tau > 0$.  There is a pruning $T^\star$ of $T$ satisfying: 
\begin{itemize} 
\item[$\circ$] $\ds\Prx_{\mathrm{uniform}~\bx}[f(\bx)\ne T^\star(\bx)]\le \tau\log s$; 
\item[$\circ$] For every node $v$ of $T^\star$, writing $i(v)$ to denote the variable queried at $v$, we have that 
\begin{equation} \Inf_{i(v)}(f_{v}) \ge \tau, \label{eq:influence-cutoff} 
\end{equation} 
where $f_v$ denotes the restriction of $f$ by the root-to-$v$ path in $T^\star$. 
\end{itemize} 
\end{theorem}

(We show in the body of this paper that this pruning lemma implies the OSSS inequality.)

For the realizable setting, this lemma is useful because only a small number of variables can satisfy~\Cref{eq:influence-cutoff}. It is well known and easy to show that the {\sl total} influence of a size-$s$ decision tree, the sum of individual variable influences, is upper bounded by $\log s$.  There can therefore be at most $(\log s)/\tau$ many variables with influence at least $\tau$.  Our $\poly(n)\cdot  s^{O(\log\log s)}$ time algorithm for the realizable setting follows quite easily from~\Cref{thm:pruning intro}.

\paragraph{The agnostic setting.}  In the agnostic setting, there is no longer a good bound on the number of variables of $f$ with influence at least $\tau$.  If $f$ is merely {\sl close to} a size-$s$ decision tree, say $0.1$-close, the size of this set can be as large as $\Omega(n)$ as opposed to $(\log s)/\tau$ as in the realizable setting.  

To overcome this, we consider the {\sl smoothing} of $f$ and the {\sl noisy influence} of its variables, and rely on a generalization of our pruning lemma based on these notions.  By choosing an appropriate smoothing/noise parameter $\delta$, we show that: 
\begin{itemize} 
\item[$\circ$] The smoothing $\tilde{f}$  is $(\delta\log s)$-close to $f$; 
\item[$\circ$] There are at most $1/(\tau\delta)$ many variables with noisy influence at least $\tau$ on $\tilde{f}$.  
\end{itemize} 
A straightforward application of these ideas yields an agnostic algorithm that achieves accuracy $O(\opt)+\eps$. A more careful analysis further improves the guarantee to $\opt+\eps$.

The high-level idea of using smoothing and noisy influence to upgrade a non-agnostic algorithm into an agnostic one already appears in prior work on decision tree learning~\cite{BGLT-NeurIPS1}, though the details of our analyses differ.

\section{Preliminaries} 
\label{sec:prelim} 
We use {\bf boldface} (e.g.~$\bx\sim\bn$) to denote random variables, and unless otherwise stated, all probabilities and expectations are with respect to the uniform distribution.  A restriction $\pi$ of a function $f : \bn \to \R$, denoted $f_\pi$, is the subfunction of $f$ that one obtains by fixing a subset of the variables to constants (i.e.~$x_i = b$ for $i\in [n]$ and $b\in \bits$).  We write $|\pi|$ to denote the number of variables fixed by $\pi$.  

The {\sl size} of a tree is its number of leaves, its {\sl depth} is the length of the longest root-to-leaf path, and we define its {\sl average depth} to be the quantity: 
    \[ \Delta(T) \coloneqq \Ex_{\bx\sim\bits^n}[\textnormal{depth of leaf that $\bx$ reaches}] = \sum_{\textnormal{leaves $\ell\in T$}} 2^{-|\ell|}\cdot |\ell|,\]
    where $|\ell|$ denotes the depth of $\ell$ within $T$.  Note that if $T$ is a size-$s$ decision tree, then $\Delta(T)\le \log s$.

\begin{definition}[Influence of variables]
    \label{def:influence}
    For $f:\bits^n \to \bits$ and $i \in [n]$, the {\sl influence of $x_i$ with respect to $f$} is the quantity
    \begin{align}  
        \label{eq:inf definition not equals}
        \Inf_i(f) \coloneqq \Prx_{\bx \sim \bits^n}\left[f(\bx) \neq f(\bx^i)\right],
    \end{align}
    where $\bx^i$ denotes $\bx$ with its $i^{\text{th}}$ coordinate rerandomized (i.e.~flipped with probability $\frac1{2}$).   More generally, for $f: \bits^n \to Y$ where $Y$ is a metric space equipped with a distance function $\rho$, 
    \begin{align}
    \label{eq:inf definition metric}
        \Inf_i(f) \coloneqq \Ex_{\bx \sim \bits^n}\left[\rho(f(\bx), f(\bx^i))\right].
    \end{align}
\end{definition}
\begin{remark}[Metric spaces of interest]  Although the focus of our work is on learning boolean-valued functions $f : \bn\to\bits$, our approach involves reasoning more generally about real-valued functions $\tilde{f} : \bn\to \R$.  Several intermediate results that we establish for real-valued functions hold even more generally for any metric space $Y$ as the codomain (e.g.~our pruning lemma), and in those cases we state and prove them in their most general form. 

Throughout this paper the codomain $Y = \bits$ is by default equipped with the not-equals metric $\rho(x,y) = \Ind[x\ne y]$ (note that in this case \Cref{eq:inf definition not equals,eq:inf definition metric} are equivalent), and the codomain $Y = \R$ is by default equipped with the absolute value metric $\rho(x,y) = |x-y|$.    
\end{remark}

\begin{definition}[Distance between functions]
    For any metric space $Y$ equipped with a distance function $\rho$, we define the distance between two functions $f,g: \bits^n \to Y$ to be
    \begin{align*}
        \dist(f,g)\coloneqq \Ex_{\bx \sim \bn}[\rho(f(\bx), g(\bx))]. 
    \end{align*}
    We say that $f$ is {\sl $\eps$-close to} $g$ if $\dist(f,g)\le \eps$.
\end{definition}

We note that we can express influence in terms of distance.
\begin{fact}
    \label{fact:influence with distance}
    For any metric space $Y$, function $f: \bits \to Y$, and $i \in [n]$,
    \begin{align*}
        \Inf_i(f) = \dist(f, f_{x_{i} = 1}) = \dist(f, f_{x_{i} = -1}).
    \end{align*}
\end{fact}




\paragraph{Fourier analysis of boolean functions.} 
We will need the very basics of the Fourier analysis of boolean functions; for an in-depth treatment, see~\cite{ODBook}.  Every function $f : \bn \to \R$ can be uniquely expressed as a multilinear polynomial via its {\sl Fourier expansion}: 
\begin{equation*} f(x) = \sum_{S\sse [n]} \hat{f}(S) \prod_{i\in S}x_i, \quad \textstyle \text{where $\hat{f}(S) = \E\big[f(\bx)\prod_{i\in S}\bx_i\big]$}. 
\end{equation*}

\begin{definition}[Smoothed version of a function]
\label{def:smooth f}
For a function $f: \bn \to \bits$ and noise rate $\delta \in [0, 1]$, the {\sl $\delta$-smoothed version of $f$} is the function $f_\delta: \bits^n \to [-1, 1]$  defined as
\[
    f_\delta(x) \coloneqq \Ex_{\tilde\bx \sim_{\delta} x}[f(\tilde \bx)],
\]
where $\tilde\bx \sim_{\delta} x$ denotes drawing $\tilde \bx$ such that each coordinate $\tilde\bx_i$ is set to $x_i$ with probability $1 - \delta$, and rerandomized with probability $\delta$. Equivalently, each bit of $x$ gets flipped in $\tilde\bx$ with probability $\frac{\delta}{2}$ independently.
\end{definition} 
We remark that $f_\delta$ is sometimes also denoted $\textnormal{T}_{1-\delta}f$, with $\textnormal{T}_{1-\delta}$ being called the noise operator with parameter $1-\delta$.
\input{Pruning}

\section{Learning in the realizable setting} 
\label{sec:realizable} 
\input{realizable}

\section{Learning monotone target functions in the agnostic setting}

In the remainder of this paper we extend our analysis from the realizable to the agnostic setting.  As alluded to in the introduction, the main challenge that arises when if $f$ is merely {\sl close to} a small-size decision tree, instead of being {\sl exactly} a small-size decision tree, is that we no longer have a good bound on the number of its variables with  influence at least $\tau$.  In the realizable setting we were able to bound this number by $(\log s)/\tau$ (\Cref{eq:not too many influential} in the proof of~\Cref{lem: main realizable}) but this crucially relied on the assumption that $f$ and and its subfunctions are size-$s$ decision trees, and hence have total influence at most $\log s$.  

The way we handle this in the case of general target functions is somewhat involved; we give the full analysis in the next section.  In this section we consider the special case of {\sl monotone} target functions and prove~\Cref{thm:monotone}.  For monotone functions $f$, we show that we can easily bound the number of variables of influence at least $\tau$ by $1/\tau^2$, even if $f$ is not a small-size decision tree.  Furthermore, we also show that for monotone targets $f$ our algorithm does not need membership queries to $f$ and can instead rely only on uniform random labeled examples.

\input{monotone}

\section{Learning general target functions in the agnostic setting}




In this section we prove~\Cref{thm:main}.  The algorithm for the agnostic setting calls the same procedure $\BuildDT$ as in the realizable setting, but on  the smoothed version $f_\delta$ of function $f$ (recall~\Cref{def:smooth f}). 

\paragraph{Correctness.} We'll prove that the output of $\BuildDT$ on $f_\delta$ is close to $f$. For that, we'll need some facts about the noise operator.
\begin{fact}[Noise sensitivity of decision trees]
    \label{fact:NS-DT}
    For any $\delta \in (0,1)$ and decision tree $T: \bn \to \bits$,
\[         \dist(T_\delta, T) \leq \Delta(T) \cdot \delta.
    \]
\end{fact}
\begin{proof}
    We expand the distance between $T_\delta$ and $T$,
    \begin{align*}
        \dist(T_\delta, T) &= \Ex_{\bx \sim \bits^n}\left[|T_\delta(\bx) - T(\bx)|\right] \\
        &= \Ex_{\bx \sim \bits^n}\left[\left|\Ex_{\tilde\bx \sim_{\delta} \bx}[T(\tilde \bx)] - T(\bx)\right|\right] \\
        &= 2 \cdot \Ex_{\bx \sim \bits^n}\left[\Prx_{\tilde\bx \sim_{\delta} \bx}[T(\tilde\bx ) \neq T(\bx)\right].
    \end{align*}
    For any $\bx \in \bits^n$, let $d(\bx)$ be the depth of the leaf in $T$ that $\bx$ reaches. In order for $T(\tilde\bx) \neq T(\bx)$, $T(\tilde\bx)$ must reach a different leaf in $T$ than $\bx$ does. For that to happen, one of the $d(\bx)$ coordinates $T$ queries for $\bx$ must flip. By union bound, this occurs with probability at most $\frac{\delta \cdot d(\bx)}{2}$. Therefore,
\[         \dist(T_\delta, T) \leq 2 \cdot \Ex_{\bx \sim \bits^n}\left[\frac{\delta \cdot d(\bx)}{2}\right] = \Delta(T) \cdot \delta.\qedhere
\] \end{proof}

\begin{fact}[Noise operator is self-adjoint, also in \cite{ODBook}]
    \label{fact:self-adjoint}
    For any functions $f, g: \bn \to \bits$,
    \begin{align*}
        \dist(f_\delta, g) = \dist(f, g_\delta). 
    \end{align*}
\end{fact}
\begin{proof}
    Drawing $\bx \sim \bits^n$ uniformly and then $\tilde\bx \sim_{\delta} \bx$ gives the same joint distribution over $(\bx, \tilde\bx)$ as first drawing $\tilde\bx \sim_{\delta} \bits^n$ uniformly and then $\bx \sim_{\delta} \tilde\bx$. That fact is used between the third and fourth line of the following series of algebraic manipulations. 
    \begin{align*}
        \dist(f_\delta, g) &= \Ex_{\bx \sim \bits^n}\left[\left|\Ex_{\tilde\bx \sim_{\delta} \bx}[f_\delta(\tilde \bx)] - g(\bx)\right|\right] \\
        &= 2 \cdot \Ex_{\bx \sim \bits^n}\left[\Prx_{\tilde\bx \sim_{\delta} \bx}[f(\tilde\bx ) \neq g(\bx)]\right] \\
        &= 2 \cdot \Ex_{\bx \sim \bits^n, \tilde\bx \sim_{\delta} \bx}\big[\Ind[f(\tilde\bx ) \neq g(\bx)]\big] \\
        &= 2 \cdot \Ex_{\tilde\bx \sim \bits^n, \bx \sim_{\delta} \bx}\big[\Ind[f(\tilde\bx ) \neq g(\bx)]\big] \\
        &= 2 \cdot \Ex_{\tilde\bx \sim \bits^n}\left[\Prx_{\bx \sim_{\delta} \bx}[f(\tilde\bx ) \neq g(\bx)]\right] \\
        &= \dist(f, g_{\delta}). \qedhere 
    \end{align*}
\end{proof}

Given the above two facts, we are able to prove that our algorithm has the desired error on guarantee.
\begin{lemma}\label{lemma:correctness on smoothed}
    For any size $s$ and $\eps \in (0,1)$, set $d \coloneqq \log(\frac{s}{\eps})$ and $\tau \coloneqq \frac{\eps}{\log s}$. Then, for any $\delta \leq \frac{\eps}{\log s}$, $\BuildDT_M(f_\delta,\varnothing,s,d,\tau)$ returns a decision tree $T$ satisfying
    \begin{align*}
        \dist(T, f) \leq \opt_s + 4  \eps.
    \end{align*}
\end{lemma}
\begin{proof}
   Let $T^\star$ be the size-$s$ decision tree that $f$ is $\opt_s$-close to. First, we show $T^\star$ is also close to $f$.
    \begin{align*}
        \dist(T^\star, f_\delta) &= \dist((T^\star)_\delta, f) \tag{\Cref{fact:self-adjoint}} \\
        &\leq \dist((T^\star)_\delta, T^\star) + \dist(T^\star ,f) \tag{Triangle inequality} \\
        &\leq \delta \cdot \Delta(T^\star) + \opt_s \tag{\Cref{fact:NS-DT}} \\
        &\leq \eps + \opt_s. \tag{$\Delta(T^\star) \leq \log(s), \delta \leq \frac{\eps}{\log s}$}
    \end{align*}
    
     By \Cref{thm:pruning general} (applied with the metric space $Y = [-1,1]$), we know that there is some $T^\star_{\pruned}$ that is everywhere $\tau$-influential with respect to $f_\delta$ satisfying,
     \begin{align*}
         \dist(T^\star_{\pruned}, f_\delta) \leq \big(\opt_s + \eps\big) + \eps = \opt_s + 2\eps.
     \end{align*}
     As in the proof of \Cref{lem: main realizable}, let $T^{\star}_{\mathrm{trunc}}$ be $T^\star_{\pruned}$ truncated to depth~$d$ (where the new leaves introduced by truncated paths are labeled with arbitrary leaf values, say $1$).  This tree $T^{\star}_{\mathrm{trunc}}$ is a depth-$d$, size-$s$, everywhere $\tau$-influential tree that satisfies
     \begin{align*}
         \dist(f_\delta,T^{\star}_{\mathrm{trunc}}) \le \dist(f_\delta,T^\star_{\pruned}) + \eps \le \opt_s + 3\eps.
     \end{align*}
     Therefore, by~\Cref{claim: BuildDT optimal} $\BuildDT$ returns a tree $T$ that also satisfies $\dist(f_\delta,T)\le \opt_s + 3\eps$. Finally, we bound the distance between $f$ and $T$.
     \begin{align*}
         \dist(f, T) &\leq \dist(f, T_\delta) + \dist(T, T_\delta) \tag{Triangle inequality} \\
         &\leq \dist(f_\delta, T) + \dist(T, T_\delta) \tag{\Cref{fact:self-adjoint}} \\
         &\leq (\opt_s + 3\eps) + \delta \cdot \Delta(T) \tag{\Cref{fact:NS-DT}} \\
         &\leq  (\opt_s + 3\eps) + \eps = \opt_s + 4\eps. \tag{$\Delta(T) \leq \log s, \delta \leq \frac{\eps}{\log s}$}
     \end{align*}
\end{proof}

\paragraph{Efficiency.} Now we analyze the runtime of the procedure $\BuildDT$ on the smoothed function $f_{\delta}$. As in the proof of \Cref{claim: runtime} for the realizable setting, we need to upper bound the number of different recursive calls to the procedure. The key step is to control the size of $S$, the set of variables that is sufficient influential (w.r.t.\ function $(f_\delta)_\pi$ and threshold $\tau$).

We start with a well-known fact stating that the total influence of any $\delta$-smoothed function is at most $O(1/\delta)$. Here, we use a slightly different version of influence that is defined as the expected squared difference between the functions values at $\bx$ and $\bx^{\sim i}$. This squared influence does not fit into \Cref{def:influence} since the squared difference is not a metric, but the advantage is that it can be easily expressed in terms of the Fourier coefficients of the function.

\begin{fact}[Total influence of smoothed functions]\label{fact:total influence smoothed}
    For any $f: \bn \to \bits$ and $\delta \in (0, 1]$,
    \[
        \sum_{i=1}^{n}\Ex_{\bx \sim \bn}\left[(f_\delta(\bx) - f_\delta(\bx^{\sim i}))^2\right] \le \frac{1}{e\delta}.
    \]
\end{fact}

\begin{proof}
    Suppose that the Fourier expansion of $f$ is $f(x) = \sum_{S \subseteq [n]}\wh f(S)\prod_{i \in S}x_i$. The Fourier coefficients of $f_\delta$ are given by $\wh{f_\delta}(S) = (1 - \delta)^{|S|}\wh f(S)$. Then, using the Fourier formula for the total squared influence,
    \begin{align*}
        \sum_{i=1}^{n}\Ex_{\bx \sim \bn}\left[(f_\delta(\bx) - f_\delta(\bx^{\sim i}))^2\right]
    &=   2\sum_{S \subseteq [n]}|S|\cdot \left[\wh{f_\delta}(S)\right]^2\\
    &=   \sum_{S \subseteq [n]}2|S|\cdot (1 - \delta)^{2|S|} \cdot \left[\wh f(S)\right]^2\\
    &\le \frac{1}{e\delta}\sum_{S \subseteq [n]}\left[\wh f(S)\right]^2\\
    &= \frac{1}{e\delta}.
    \end{align*}
    The third step applies $\max_{x \ge 0}x(1-\delta)^x \le \max_{x \ge 0}xe^{-\delta x} = 1/(e\delta)$,
    and the last step applies Parseval's identity (\Cref{fact:Parseval}). 
\end{proof}

For any restriction $\pi$, applying \Cref{fact:total influence smoothed} to $f_\pi$ allows us to control the number of variables that have large influences w.r.t.\ $(f_\pi)_\delta$. To upper bound the runtime of $\BuildDT$, however, we need a similar guarantee for the function $(f_\delta)_\pi$, which is different from $(f_\pi)_\delta$ in general. Fortunately, the following fact states that for small $\delta$, the two functions are pointwise close, and thus allows us to relate the influences of each variable $x_i$ w.r.t.\ the two functions.

\begin{fact}\label{fact:similar infuence}
    For any $f: \bn \to \bits$ and restriction $\pi$, it holds for every $x \in \bn$ that
    \[
        |(f_\pi)_\delta(x) - (f_\delta)_\pi(x)| \le \delta|\pi|.
    \]
    Furthermore, for every $i \in [n]$,
    \[
        \left|\Inf_i((f_\pi)_\delta) - \Inf_i((f_\delta)_\pi)\right| \le 2\delta|\pi|.
    \]
\end{fact}

\begin{proof}
    Fix $x \in \bn$ and consider the following procedure for calculating $(f_\pi)_\delta(x)$:
    \begin{enumerate}
        \item Set $\by \gets x$ and draw $s_1, s_2, \ldots, s_n$ independently from $\mathrm{Beroulli}(\delta / 2)$.
        \item For each $i \in [n]$, negate $\by_i$ if $s_i = 1$.
        \item For each constraint ``$x_i = b$'' in $\pi$, set the $i$-th bit of $\by$ to $b$.
    \end{enumerate}
    We can easily verify that $(f_\pi)_\delta(x) = \Ex[f(\by)]$, where the expectation is over the randomness in $s$.
    
    Furthermore, $(f_\delta)_\pi(x)$ can be defined by an almost identical procedure, with Steps 2~and~3 performed in reverse order: We start with $\bz = x$ and draw $s \in \bits^n$ randomly. We set $\bz_i$ to $b$ for each constraint ``$x_i = b$'' in $\pi$, and then negate $\bz$ according to the non-zero entries in $s$. Similarly, we have $(f_\delta)_\pi(x) = \Ex[f(\bz)]$.
    
    We can couple the two procedures by sharing the random bits $s_1$ through $s_n$. Note that if $s_i = 0$ holds for every index $i$ that appears in $\pi$, we would end up with $\by = \bz$. In other words, $\by$ and $\bz$ may differ only when $s_i = 1$ for some index $i$ that appears in $\pi$, which, by a union bound, happens with probability $\le |\pi| \cdot (\delta/2)$. Since $f$ has codomain $\bits$, we have
    \[
        |(f_\pi)_\delta(x) - (f_\delta)_\pi(x)|
    =   \left|\Ex[f(\by)] - \Ex[f(\bz)]\right|
    \le \Ex\left[|f(\by) - f(\bz)|\right]
    \le   2\Pr[\by \ne \bz]
    \le \delta|\pi|,
    \]
    where the probability and expectations are over the coupling of $(\by, \bz)$ defined earlier.
    
    The second part of the fact follows immediately: the first part implies
    \[
        |(f_\pi)_\delta(x) - (f_\pi)_\delta(y)|
    -   |(f_\delta)_\pi(x) - (f_\delta)_\pi(y)|
    \in [-2\delta|\pi|, 2\delta|\pi|]
    \]
    for every $x, y \in \bn$. Therefore, the difference between the influences,
    \[
        \Inf_i((f_\pi)_\delta) - \Inf_i((f_\delta)_\pi)
    =   \Ex_{\bx\sim\bn}\left[|(f_\pi)_\delta(\bx) - (f_\pi)_\delta(\bx^{\sim i})| - |(f_\delta)_\pi(\bx) - (f_\delta)_\pi(\bx^{\sim i})|\right],
    \]
    is also in $[-2\delta|\pi|, 2\delta|\pi|]$.
\end{proof}

\begin{claim}[Runtime]
\label{claim:runtime on smoothed}
For all $d, s \in \N$ and $\tau, \delta > 0$ that satisfy $\tau > 2\delta d$, assuming that variable influences of $f_\delta$ and its subfunctions can be computed exactly in unit time,
the algorithm $\BuildDT_M({f_\delta}, \varnothing, s, d, \tau)$ runs in time
\[
n \cdot \poly(s) \cdot \left(\frac{1}{e\delta(\tau - 2\delta d)^2}\right)^{O(d)}.
\]
In particular, for $\delta = \tau/(4d)$, the runtime is $n\cdot \poly(s) \cdot (d/\tau)^{O(d)}$. 
\end{claim}

\begin{proof}
    As in the proof of \Cref{claim: runtime}, it suffices to show that when invoking $\BuildDT_M(f_\delta, \varnothing, s, d, \tau)$, at most $s \cdot \left(\frac{1}{e\delta(\tau - 2\delta d)^2}\right)^{O(d)}$ different parameter tuples are passed to the recursive calls. It is, in turn, sufficient to prove that $|S| \le \frac{1}{e\delta(\tau - 2\delta d)^2}$ holds for every recursive call $\BuildDT_M(f_\delta, \pi, s', d, \tau)$, where $S$ is the set of indices $i$ that satisfy $\Inf_i((f_\delta)_\pi) \ge \tau$. We note that
    \begin{align*}
        i \in S
    &\iff   \Inf_i((f_\delta)_\pi) \ge \tau\\
    &\implies \Inf_i((f_\pi)_\delta) \ge \tau - 2\delta d \tag{\Cref{fact:similar infuence} and $|\pi| \le d$}\\
    &\iff \Ex_{\bx \sim \bn}\left[|(f_\pi)_\delta(\bx) - (f_\pi)_\delta(\bx^{\sim i})|\right] \ge \tau - 2\delta d \tag{definition of influence}\\
    &\implies \Ex_{\bx \sim \bn}\left[((f_\pi)_\delta(\bx) - (f_\pi)_\delta(\bx^{\sim i}))^2\right] \ge (\tau - 2\delta d)^2. \tag{Jensen's inequality and $\tau > 2\delta d$}
    \end{align*}
    Applying \Cref{fact:total influence smoothed} to $f_\pi$ shows that the above can hold for at most $\frac{1/(e\delta)}{(\tau - 2\delta d)^2}$ different indices $i$. This proves $|S| \le \frac{1}{e\delta(\tau - 2\delta d)^2}$ and finishes the proof.
\end{proof}

Now we put everything together to prove our main theorem.

\begin{proof}[Proof of \Cref{thm:main}]
Let $d\coloneqq \log(s/\eps)$, $\tau \coloneqq \frac{\eps}{\log s}$ and $\delta \coloneqq \frac{\tau}{4d}$. By~\Cref{lemma:correctness on smoothed}, $\BuildDT_M(f_\delta, \varnothing, s, d, \tau)$ returns a decision tree $T$ that satisfies $\dist(T, f)\le \opt_s + 4\eps$.

For the runtime, in~\Cref{claim:runtime on smoothed} we again assumed that the influences of $f_\delta$ and its restrictions can be computed exactly in unit time. As in the proof of \Cref{lem: main realizable}, estimating these influences up to an $O(\tau)$ additive error would suffice. Given query access to $f$, $(f_\delta)_\pi(x)$ can be estimated up to $O(\tau)$ error with probability $1 - \delta$ using $O(\log(1/\delta)/\tau^2)$ queries for any restriction $\pi$ and input $x$. Then, by randomly sampling $O(\log(1/\delta)/\tau^2)$ copies of $\bx \sim \bits$ and estimating both $(f_\delta)_\pi(\bx)$ and $(f_\delta)_\pi(\bx^{\oplus i})$, we can estimate $\Inf_i((f_\delta)_\pi)$ up to $O(\tau)$ error with probability $1 - O(\log(1/\delta)/\tau^2)\cdot \delta$.

By~\Cref{claim:runtime on smoothed}, the number of variable influences that need to be computed is at most $n\cdot(d/\tau)^{O(d)}$. By setting $\delta < 1/(n^2\cdot (d/\tau)^{O(d)})$, we can ensure that
\[
    n\cdot(d/\tau)^{O(d)} \cdot O(\log(1/\delta)/\tau^2)\cdot \delta \ll 1,
\]
so that w.h.p.\ all the influence estimates are accurate up to $O(\tau)$ error. Note that estimating each influence takes $[O(\log(1/\delta)/\tau^2)]^2$ queries and thus runs in time
\[
    n \cdot [O(\log(1/\delta)/\tau^2)]^2
=   \tilde O(n) \cdot \poly(d/\tau).
\]
Together with~\Cref{claim:runtime on smoothed}, this upper bounds the overall runtime of the algorithm by
\[
n \cdot \poly(s) \cdot \left(\frac{d}{\tau}\right)^{O(d)} \cdot \tilde O(n) \cdot \poly(d/\tau)
\le   \tilde O(n^2) \cdot (s/\eps)^{O(\log((\log s)/\eps))}. \qedhere 
\]
\end{proof} 


\section{Conclusion} 

We have given an $n^{O(\log\log n)}$-time membership query algorithm for properly learning decision trees under the uniform distribution, improving on the previous fastest runtime of $n^{O(\log n)}$.  The obvious open problem is to obtain a polynomial-time algorithm, which would bring the state of the art for proper learning of decision trees into alignment with that of improper learning~\cite{KM93,GKK08}.

Improved learning algorithms for decision trees often go hand in hand with an improved understanding of their structure.  Ehrenfeucht and Haussler's algorithm~\cite{EH89} is based on the observation that one of the subtrees of the root of a size-$s$ decision tree has size $\le s/2$;~\cite{BLT-ITCS} uses the OSSS inequality to show that influence is a good proxy for quality as a root; our algorithm is built on our decision tree pruning lemma, which strengthens the OSSS inequality and the connection between influence and root quality.  A natural next step is to formulate and develop new structural results that will facilitate a polynomial-time algorithm.

 Concluding on a speculative note, we remark that~\cite{BLT-ITCS}'s algorithm is modeled after practical heuristics, such as ID3, CART, and C4.5, for learning decision trees.  These are some of the earliest and most basic algorithms in machine learning, and they continue to be widely used to this day. Our algorithm extends~\cite{BLT-ITCS}'s and circumvents lower bounds that~\cite{BLT-ITCS} had established for their algorithm.  It would be interesting to explore possible practical implications of our work.   
\section*{Acknowledgements}

We are grateful to the anonymous reviewers, whose comments and suggestions have helped improve this paper.

Guy and Li-Yang are supported by NSF CAREER Award 1942123. Mingda is supported by DOE Award DE-SC0019205 and ONR Young Investigator Award N00014-18-1-2295. Jane is supported by NSF Award CCF-2006664.

%% file: Pruning.tex
\section{Our pruning lemma}  
\label{sec:prune} 
In this section we prove our key new structural result, the decision tree pruning lemma.  The actual result that we establish,~\Cref{thm:pruning general}, generalizes the pruning lemma as stated in the introduction (\Cref{thm:pruning intro}) in two ways: 
\begin{enumerate}
\item It holds for functions mapping into an arbitrary metric space rather than just boolean-valued functions; 
\item The decision tree $T$ need not compute $f$.   
\end{enumerate} 
Both aspects will be needed for the application to agnostic learning. 


\begin{definition}[Everywhere $\tau$-influential]
    \label{def:local influence}
    For any function $f: \bits^n \to Y$, threshold $\tau > 0$, and decision tree $T: \bits^n \to Y$, we say that $T$ is {\sl everywhere $\tau$-influential with respect to $f$} if, for every internal node $v$ of $T$, writing $i(v)$ to denote the variable queried at $v$, we have 
    \begin{align*}
        \Inf_{i(v)}(f_v) \ge \tau,
    \end{align*}
    where $f_v$ denotes the restriction of $f$ by the root-to-$v$ path in $T$.
\end{definition}

The proof of our pruning lemma is constructive---we give an efficient algorithm (\Cref{fig:Prune}) showing how to prune $T$ so that the resulting tree is everywhere $\tau$-influential with respect to~$f$---though in our applications to learning we do not need it to be constructive.

\begin{figure}[h]
  \captionsetup{width=.9\linewidth}
\begin{tcolorbox}[colback = white,arc=1mm, boxrule=0.25mm]
\vspace{3pt} 

$\Prune(f, T, \tau)$:

\begin{itemize}[align=left]
    \item[\textbf{Input:}] Query access to a function $f: \bits^n \to Y$, a decision tree $T: \bits^n \to Y$, and a threshold $\tau > 0$.
    \item[\textbf{Output:}] A decision tree that is everywhere $\tau$-influential w.r.t. $f$.
\end{itemize}
\begin{enumerate}
    \item If $T$ has depth $0$, return $T$.
    \item Let $x_i$ be the variable queried at root of $T$ and $T_{-1}$ and $T_1$ be its left and right subtree respectively.
    \item If $\Inf_i(f) > \tau$, return the tree that queries $x_i$ as its root and has $\Prune(f_{x_i = -1}, T_{-1}, \tau)$ and $\Prune(f_{x_i = 1}, T_1, \tau)$ as its left and right subtree respectively.
    \item If $\Inf_i(f) \leq \tau$, return whichever of $\Prune(f, T_{-1}, \tau)$ or $\Prune(f, T_1, \tau)$ have less distance w.r.t. $f$.
\end{enumerate}

\end{tcolorbox}
\caption{A procedure for pruning a tree $T$ to ensure it is everywhere $\tau$-influential with respect to a function $f$.}
\label{fig:Prune}
\end{figure}

The remainder of this subsection will be devoted to proving the following generalization of \Cref{thm:pruning intro}.
\begin{theorem}[Properties of $\Prune$]
    \label{thm:pruning general}
    For any metric space $Y$, function $f: \bits \to Y$, decision tree $T: \bits \to Y$, and threshold $\tau > 0$, let $T^\star = \Prune(f, T, \tau)$. Then,
    \begin{enumerate}
        \item {\sl Size and depth do not increase:} The size and depth of $T^\star$ are at most the size and depth of~$T$.
        \item {\sl Everywhere $\tau$-influential:} $T^\star$ is everywhere $\tau$-influential with respect to $f$.
        \item {\sl Small increase in distance:} For $\Delta(T)$ the average depth of $T$, 
        \begin{align*}
            \dist(T^\star, f) \leq \dist(T, f) + \Delta(T) \cdot \tau.
        \end{align*}
    \end{enumerate}
\end{theorem}
\Cref{thm:pruning intro} is a special case of \Cref{thm:pruning general} where $Y = \bits$ with the not-equals metric and $f \equiv T$. (Recall also that $\Delta(T) \le \log s$.) We prove each guarantee of \Cref{thm:pruning general} separately.

\begin{proof}[Proof of the first guarantee of~\Cref{thm:pruning general}]
    By induction on the depth of $T$. If $T$ has depth $0$, then the size and depth of $\Prune(f, T, \tau)$ are the same as $T$. For $x_i$ the variable queried at root of $T$ and $T_{-1},T_1$ its left and right subtrees respectively, if $\Inf_i(f) > \tau$,
    \begin{align*}
        \size(\Prune(f, T, \tau)) &= \size(\Prune(f_{x_i = {-1}}, T_{-1}, \tau)) + \size(\Prune(f_{x_i = 1}, T_1, \tau)) \\
        &\leq \size(T_{-1}) + \size(T_1) \\
        &= \size(T).
    \end{align*}
    where the second step is the inductive hypothesis. Similarly, for depth
    \begin{align*}
        \depth(\Prune(f, T, \tau)) &= 1 + \max\big(\depth(\Prune(f_{x_i = {-1}}, T_{-1}, \tau)), \depth(\Prune(f_{x_i = 1}, T_1, \tau))\big) \\
        &\leq 1 + \max\big(\depth( T_{-1},), \depth( T_1)\big) \\
        &= \depth(T).
    \end{align*}
    Finally, if $\Inf_i(f) \leq \tau$, then $\Prune(f, T, \tau)$ is equal to either $\Prune(f, T_{-1}, \tau)$ or $\Prune(f, T_1, \tau)$. Since $T_{-1}$ and $T_1$ each have size and depth less than those of $T$, the desired result holds by the inductive hypothesis.
\end{proof}

\begin{proof}[Proof of the second guarantee of~\Cref{thm:pruning general}]
    By induction on the depth of $T$. If $T$ has depth $0$, then it has no internal nodes, so vacuously is everywhere $\tau$-influential. Otherwise, let $x_i$ be the variable queried at root of $T$ and $T_{-1},T_1$ be its left and right subtrees respectively. If $\Inf_i(f) \leq \tau$, then $\Prune(f, T, \tau)$ is either $\Prune(f, T_{-1}, \tau)$ or $\Prune(f, T_1, \tau)$, which is everywhere $\tau$-influential w.r.t. $f$ by the inductive hypothesis.
    
    If we fell in neither of the above two cases, we have $\Inf_i > \tau$. Let $v$ be some internal node of $\Prune(f, T_{-1}, \tau)$. If $v$ is the root of $\Prune(f, T_{-1}, \tau)$, then
    \begin{align*}
        \Inf_{i(v)}(f_v) = \Inf_i(f) > \tau.
    \end{align*}
    Otherwise, let $\alpha$ be the restriction corresponding to the root-to-$v$ path. Since $x_i$ is the root of $\Prune(f, T_{-1}, \tau)$, we have that $\alpha$ must fix $x_i = b$ for $b \in \bits$ and $v$ is an internal node for $\Prune(f_{x_i = b}, T_b, \tau)$. Applying the inductive hypothesis to $\Prune(f_{x_i = b}, T_b, \tau)$, we have that $\Inf_{i(v)}(f_v) > \tau$.
\end{proof}

Before we prove the third and final guarantee of \Cref{thm:pruning general}, we state two easy facts about the subtrees of a decision tree.

\begin{fact}[Subtrees of a tree]
    Let $T: \bits^n \to Y$ be some decision tree and $T_{-1}, T_{1}$ be its left and right subtrees respectively. Then,
    \begin{align}
        \label{eq:avg depth splits}
        \lfrac{1}{2}\cdot (\Delta(T_{-1}) + \Delta(T_{1})) = \Delta(T) - 1.
    \end{align}
    Furthermore, any function $f:\bits^n \to Y$ and $x_i$ being the root of $T$,
    \begin{align}
        \label{eq:error splits}
        \lfrac{1}{2}\cdot (\dist(T_{-1}, f_{x_i = -1}) + \dist(T_{1}, f_{x_i = 1}) ) = \dist(T, f).
    \end{align}
\end{fact}

\begin{proof}[Proof of the third guarantee of~\Cref{thm:pruning general}]
    By induction on the depth of $T$. If $T$ has depth $0$ then the claim easily holds with equality.  Otherwise, let $x_i$ be the variable queried at root of $T$ and $T_{-1},T_1$ be its left and right subtrees respectively. We note that the depth of $T_{-1}$ and $T_1$ are strictly less than the depth of $T$, so we can apply our inductive hypothesis to them. We consider two cases.
    
    \pparagraph{Case 1:} $\Inf_i(f) > \tau$. 
    \begin{align*}
        \dist(\Prune(f, T, \tau), f) &= \frac{1}{2}\cdot \left(\sum_{b \in \bits} \dist(\Prune(f_{x_i = b}, T_{b}, \tau), f_{x_i = b}) \right) \tag{\Cref{eq:error splits}} \\
         &\leq \frac{1}{2}\cdot \left(\sum_{b \in \bits} \dist(T_b, f_{x_i = b}) + \Delta(T_{b}) \cdot \tau \right) \tag{Inductive hypothesis} \\
         &= \dist(T, f) + (\Delta(T) - 1) \cdot \tau \tag{\Cref{eq:avg depth splits,eq:error splits}} \\
         &\leq \dist(T, f) + \Delta(T) \cdot \tau.
    \end{align*}
    
    \pparagraph{Case 2:} $\Inf_i(f) \leq \tau$. 
    \begin{align*}
        \dist(\Prune(f, T, \tau), f) &= \min_{b \in \bits}\left\{\dist(\Prune(f, T_{b}, \tau), f)\right\}\\
        &\leq \frac{1}{2} \cdot \left( \sum_{b \in \bits}\dist(\Prune(f, T_{b}, \tau), f)\right)\tag{min $\leq$ average}\\
        &\leq \frac{1}{2} \cdot \left(  \sum_{b \in \bits}\dist(T_b, f) + \Delta(T_{b}) \cdot \tau\right) \tag{Inductive hypothesis} \\
        &\leq \frac{1}{2} \cdot \left(  \sum_{b \in \bits}\dist(T_b, f_{x_i = b}) + \dist(f, f_{x_i = b}) + \Delta(T_{b}) \cdot \tau\right) \tag{Triangle inequality} \\
        &= \frac{1}{2} \cdot \left(  \sum_{b \in \bits}\dist(T_b, f_{x_i = b}) + \Inf_i(f) + \Delta(T_{b}) \cdot \tau\right) \tag{\Cref{fact:influence with distance}} \\
        &= \dist(T, f) + (\Delta(T) - 1) \cdot \tau + \Inf_i(f) \tag{\Cref{eq:avg depth splits,eq:error splits}}  \\
        & \leq \dist(T,f) + \Delta(T) \cdot \tau. \tag{$\Inf_i(f) \leq \tau$}
    \end{align*}
    This completes the proof.
\end{proof}



\subsection{Our pruning lemma implies the OSSS inequality}

Several variants of the OSSS inequality (\Cref{thm:OSSS}) have been proved over the years~\cite{Lee10,JZ11,ODBook,DRT19}.  
We show that our pruning lemma implies the following strengthening of the  OSSS inequality:
\begin{theorem}[\cite{JZ11}]
    \label{thm:JZ OSSS}
    For any function $f: \bits^n \to \bits$ and decision tree $T: \bits^n \to \bits$,
    \begin{align*}
        \max_{i \in [n]} \big\{\Inf_i(f)\big\} \geq \frac{\bias(f) - \dist(T,f)}{\Delta(T)}
    \end{align*}
    where the bias of $f$ is defined as 
    \begin{align*}
        \bias(f) = \min_{b \in \{\pm 1\}} \left\{\Prx_{\bx \sim \bits^n}[f(\bx) \neq b]\right\}.
    \end{align*}
\end{theorem}
The OSSS inequality follows from \Cref{thm:JZ OSSS} by taking $T = f$, and because $2\cdot \bias(f) \leq \Var(f)$.
We now show that \Cref{thm:JZ OSSS} is a special case of~\Cref{thm:pruning general}:
\begin{proof}[Proof of $\Cref{thm:pruning general} \implies \Cref{thm:JZ OSSS}$]
    Set \violet{$Y = \bits$ and } $\tau= \max_{i \in [n]}\{ \Inf_i(f)\}$. There are no variables with influence  more than $\tau$ on $f$ so the only decision trees that are everywhere $\tau$-influential w.r.t.~$f$ are the trivial ones that make no queries. In other words, $\Prune(T, f, \tau)$ is either the constant $+1$ function or constant $-1$ function. Therefore,
    \begin{align*}
        \dist(f, \Prune(f, T, \tau)) \geq \bias(f).
    \end{align*}
    By the third guarantee of~\Cref{thm:pruning general},
    \begin{align*}
        \bias(f) \leq \dist(f, \Prune(f, T, \tau)) &\leq \dist(T, f) + \Delta(T) \cdot \tau \\
        &= \dist(T, f) + \Delta(T) \cdot \max_{i \in [n]} \big\{\Inf_i(f)\big\}.
    \end{align*}
 Rearranging completes the proof.    
\end{proof}

%% file: realizable.tex
We  first present and analyze our algorithm in the simpler realizable setting  where $f$ is exactly a size-$s$ decision tree (i.e., $\opt_s = 0$). 

\begin{theorem}[Special case of \Cref{thm:main}: the realizable setting]
\label{lem: main realizable}
There is an algorithm which, given as input $\eps > 0$, $s\in \N$, and query access to a size-$s$ decision tree $f : \bits^n\to\bits$, runs in time 
\[ \tilde{O}(n^2) \cdot (s/\eps)^{O(\log((\log s)/\eps))} \]
and outputs a size-$s$ decision tree hypothesis $T$ that w.h.p.~satisfies $\dist(T,f) \le  \eps$. 
\end{theorem}

For clarity, we describe our algorithm, $\BuildDT$ in~\Cref{fig:BuildDT}, under the assumption that variable influences of $f$ and its subfunctions (i.e.~the quantities $\Inf_i(f_\pi)$ for all $i$ and $\pi$) can be computed exactly in unit time.  In actuality one can only obtain high-accuracy estimates of these quantities via random sampling.  When we prove~\Cref{lem: main realizable} we will show how this assumption can be removed via standard arguments. 

\begin{figure}[ht] 
  \captionsetup{width=.9\linewidth}

\begin{tcolorbox}[colback = white,arc=1mm, boxrule=0.25mm]
\vspace{3pt} 

$\BuildDT_{M}(f, \pi, s, d, \tau)$:

\begin{itemize}[align=left]
    \item[\textbf{Input:}] Query access to a function $f: \bits^n \to [-1, 1]$, restriction $\pi$, size parameter $s$, depth parameter $d$, influence parameter $\tau$. \violet{It maintains a map $M : \{ \textnormal{restrictions}\} \times [s] \to \{ \textnormal{decision trees}\}$.}
    \item[\textbf{Output:}]A decision tree $T$ that minimizes $\dist(T, f_\pi)$ among all depth-\violet{$(d-|\pi|)$}, size-$s$, everywhere $\tau$-influential trees.
\end{itemize}
\begin{enumerate}
    \item If $\pi = \varnothing$, initialize $M$ to the empty map.
    \item If $M[\pi, s]$ is nonempty, return $M[\pi, s]$.
    \item If $|\pi| = d$ or $s  = 1$, \violet{return the singleton leaf labeled $\sign(\E[f_\pi])$.} 
    \item Otherwise: 
    \begin{enumerate}
        \item Let $S \subseteq [n]$ be the set of variables $i$ such that $\Inf_i(f_\pi) \ge \tau$.
        \item For each $i \in S$ and $k \in [s-1]$, let $T_{i,k}$ be the tree such that 
            \begin{align*}
                \mathrm{root}(T_{i,k}) &= x_i \\
                \textnormal{left-subtree}(T_{i,k}) &= \BuildDT_{M} (f, \pi \cup \{x_i = -1\}, k, d, \tau) \\
                \textnormal{right-subtree}(T_{i,k}) &= \BuildDT_{M}(f, \pi \cup \{ x_i = 1\}, s-k, d, \tau)
            \end{align*}
        \item Set \violet{$M[\pi,s]$ to be the tree among the $T_{i,k}$'s defined above with minimal distance to $f_\pi$.} 
        \item Return $M[\pi, s]$.
    \end{enumerate}
\end{enumerate}
\end{tcolorbox}

\caption{$\BuildDT$ uses dynamic programming to find the size-$s$, depth-$d$, everywhere $\tau$-influential tree of minimal distance to $f$.}
\label{fig:BuildDT}
\end{figure}

\begin{claim}[Correctness]
\label{claim: BuildDT optimal}
During the execution of $\BuildDT$, for any $f: \bn \to [-1, 1]$, restriction $\pi$, and $d,s \in \N$, if $M[\pi, s]$ is nonempty, it contains a tree $T$ that minimizes $\dist(f_\pi, T)$ among all depth-$(d - |\pi|)$, size-$s$, everywhere $\tau$-influential trees.
\end{claim}

\begin{proof}
We proceed by induction on $d - |\pi|$. When $d = |\pi|$, $\BuildDT$ populates $M[\pi, s]$ with the singleton leaf $b\in \bits$ that minimizes $\dist(f_\pi,b)$, which is indeed $\sign(\E[f_\pi])$. 
For the inductive step, note that each $T_{i,k}$ satisfies 
\[ \dist(f_\pi, T_{i,k}) = \lfrac{1}{2}\big(\dist(f_{\pi \cup \{ x_i = -1\}}, M[\pi \cup \{ x_i = -1\}, k]) + \dist(f_{\pi \cup \{ x_i = 1\} }, M[\pi \cup \{ x_i = 1\}, s-k])\big). \]   It  follows from the inductive hypothesis that $T_{i,k}$ minimizes distance among all everywhere $\tau$-influential, depth-$(d - |\pi|)$, size-$s$ trees with $x_i$ as the root, and whose left and right subtrees have sizes $k$ and $s-k$ respectively. Since $M[\pi, s]$ is chosen to minimize distance among all such $T_{i,k}$, its distance is minimal among all size-$s$, depth-$(d - |\pi|)$, everywhere $\tau$-influential trees. 
\end{proof}

\begin{claim}[Runtime]
\label{claim: runtime}
\violet{Let $d,s\in \N$.  Let $f : \bn\to\bits$ be a size-$s$ decision tree, and assume that variable influences of $f$ and its subfunctions can be computed exactly in unit time.} The algorithm $\BuildDT_M(f, \varnothing, s, d, \tau)$ runs in time  $n \cdot s^2\cdot  ((\log s)/\tau)^{O(d)}$.
\end{claim}

\begin{proof}
For all $\pi$, the size of the set $S$ defined on Step 4(a) is at most 
\begin{equation} |S| \le  \frac1{\tau} \sum_{i=1}^n \Inf_i(f_\pi) \le \frac{\log s}{\tau}, \label{eq:not too many influential} \end{equation}  
where the second inequality uses the fact that for any size-$s$ decision tree $T : \bn\to\bits$, 
\[ \sum_{i=1}^n \Inf_i(T) \le \sum_{i=1}^n \Pr[\text{$T$ queries $\bx_i$}] = \Delta(T) \le \log s.  \] 
Since $\BuildDT$ terminates once $|\pi| = d$ (Step 3), and a restriction $\pi$ is extended by $\{ x_i = b\}$ for some $b\in \bits$ only if $\Inf_i(f_\pi) \ge \tau$ (Step 4), the number of different restrictions that can be constructed  throughout the execution of the algorithm is at most 
\[ \sum_{k=1}^d \left(\frac{\log s}{\tau}\right)^k = \left(\frac{\log s}{\tau}\right)^{O(d)}. \] 
Since $\BuildDT$ returns at Step 2 if $M[\pi,s]$ is nonempty, this ensures that Step 4, the recursive part of $\BuildDT$, is reached at most once for each restriction $\pi$ and size $s$.  The total number of recursive calls is therefore upper bounded by 
\begin{equation}  s \cdot \left(\frac{\log s}{\tau}\right)^{O(d)}.  \label{eq:num recursive calls}
\end{equation} 
Outside of the recursive calls, the runtime of $\BuildDT$ is 
\begin{equation} O(n + s\cdot |S|) \le O(n) + s \cdot \left(\frac{\log s}{\tau}\right). \label{eq:runtime mod recursive calls}
\end{equation} 
The factor of $n$ comes from computing and comparing influences of variables (Line 4(a)), and the factor of $s\cdot |S|$ comes from Line 4(b), the number of different $(\textnormal{candidate root},\textnormal{size split})$ pairs.  The overall runtime is therefore at most the product of the bounds in~\Cref{eq:num recursive calls,eq:runtime mod recursive calls}, and the proof is complete. 
\end{proof} 

\begin{proof}[Proof of \Cref{lem: main realizable}]
Let $d\coloneqq \log(s/\eps)$ and $\tau \coloneqq \eps/\log s$.  We first establish correctness: we claim that $\BuildDT_M(f,\varnothing,s,d,\tau)$  returns a size-$s$ tree $T$ satisfying $\dist(T,f) \le 2\eps$.  Since $\Delta(f)\le \log s$, our pruning lemma,~\Cref{thm:pruning general}, tells us that there is a pruning $T^\star$ of $f$ that is everywhere $\tau$-influential and satisfies $\dist(f,T^\star) \le \Delta(f)\cdot \tau \le \eps$.  Let $T^{\star}_{\mathrm{trunc}}$ be $T^\star$ truncated to depth~$d$ (where the new leaves introduced by truncated paths are labeled with arbitrary leaf values, say $1$).  This tree $T^{\star}_{\mathrm{trunc}}$ is a depth-$d$, size-$s$, everywhere $\tau$-influential tree that satisfies $\dist(f,T^{\star}_{\mathrm{trunc}}) \le \dist(f,T^\star) + \eps \le 2\eps$.  Therefore, by~\Cref{claim: BuildDT optimal} $\BuildDT$ returns a tree $T$ that also satisfies $\dist(f,T)\le 2\eps$.


As for runtime, in~\Cref{claim: runtime}  we assumed that variable influences can be computed exactly in unit time, whereas in  actuality, we can only obtain estimates of these quantities via random sampling.  By inspection of our proofs, it is straightforward to verify that it suffices for these estimates to be accurate to $\pm \frac{\tau}{2}$. Query access to $f$ provides us with query access to $f_\pi$ for any $\pi$, and hence by the Chernoff bound, we can estimate $\Inf_i(f_\pi)$ to accuracy $\pm \frac{\tau}{2}$ and with confidence $1-\delta$ using $O(\log(1/\delta)/\tau^2)$ queries and in $n\cdot O(\log(1/\delta)/\tau^2)$ time.  As shown in~\Cref{claim: runtime}, the number of times variables influences are computed throughout the execution of the algorithm is at most $n\cdot ((\log s)/\tau)^{O(d)}$, and so by setting $\delta < 1/(n\cdot ((\log s)/\tau)^{O(d)})$, we ensure that w.h.p.~all our estimates are indeed accurate to within $\pm \frac{\tau}{2}$.  Combining this with~\Cref{claim: runtime}, the overall runtime of our algorithm is  
\[ n \cdot s^2 \cdot \left(\frac{\log s}{\tau}\right)^{O(d)} \cdot \frac{n}{\tau^2} \big(\log n + d\log\left((\log s)/\tau\right)\big) \le \tilde{O}(n^2) \cdot (s/\eps)^{O(\log((\log s)/\eps))}, \]
and this completes the proof. 
\end{proof}


%% file: monotone.tex

We recall two basic facts from the Fourier analysis of boolean functions:  

\begin{fact}[Parseval's identity]\label{fact:Parseval}
    For all boolean functions $f : \bits^n \to \bits$,
    \[\sum_{S \subseteq [n]} \hat{f}(S)^2 = \E[f(\bx)^2] = 1.\]
\end{fact}

\begin{fact}[Influence = linear Fourier coefficient for monotone $f$]
\label{fact: influence correlation}
    For all monotone boolean functions $f : \bits^n \to \bits$ and all $i \in [n]$, 
    \[\Inf_i(f) = \lfrac{1}{2}\E[f(\bx)\bx_i] = \lfrac{1}{2}\hat{f}(\{i\}).\]
\end{fact}

Combining these facts, we also have the following, which is needed for our runtime bound. 

\begin{corollary}
\label{cor: monotone influential vars}
     For all monotone boolean functions $f : \bits^n \to \bits$ and all $\tau \in [0,1]$,
     \[|\{i~|~ \Inf_i(f) \ge \tau\}| \le \frac{1}{4\tau^2}.\]
\end{corollary}

\begin{proof}
The sum of squares of linear Fourier coefficients is at most the sum of squares of all Fourier coefficients, so by Parseval's identity, it is at most 1: 
\[ \sum_{i \in [n]}\Inf_i(f)^2 = \lfrac{1}{4}\sum_{i=1}^n \hat{f}(i)^2 \le \lfrac{1}{4}\sum_{S\sse [n]}\hat{f}(S)^2 = \lfrac{1}{4}.  \] 
The corollary follows since $\Inf_i(f) \ge \tau$ iff  $\Inf_i(f)^2 \ge \tau^2$.
\end{proof}

\subsection{Proof of \Cref{thm:monotone}}

\begin{theorem}[\Cref{thm:monotone} restated]
\label{lem: analysis monotone}
    Let $f : \bn\to\bits$ be a monotone boolean function that is $\opt_s$-close to a size-$s$ decision tree.  Then for $d := \log (s / \eps)$  and $\tau := \eps / \log s$, the algorithm $\BuildDT_M(f, \varnothing, s, d, \tau)$ runs in time  \[ \tilde{O}(n^2) \cdot (s/\eps)^{O(\log((\log s)/\eps))},\]  uses $\poly(s/\eps)\cdot \log n$ uniform random examples labeled by $f$, and outputs a size-$s$ decision tree hypothesis $T$ that satisfies $\dist(f,T) \le \opt_s + \eps$. 
\end{theorem}
\smallskip 

The proof is very similar to that of~\Cref{lem: main realizable} and we point out the essential differences. 
\vspace{-5pt} 
\paragraph{Correctness.} \Cref{claim: BuildDT optimal} does not use the assumption that $f$ is exactly a size-$s$ decision tree, so correctness essentially follows from~\Cref{claim: BuildDT optimal} exactly as in the proof of~\Cref{lem: main realizable}.  Let $T_\opt$ be the size-$s$ decision tree that is $\opt_s$-close to $f$.  Our pruning lemma, \Cref{thm:pruning general}, tells us that there is a pruning $T^\star$ of $T_\opt$ that is everywhere $\tau$-influential and satisfies $\dist(f, T^\star) \le \opt_s + \eps$. Then, letting $T^{\star}_{\mathrm{trunc}}$ be $T^\star$ truncated to depth~$d$, we have that $T^{\star}_{\mathrm{trunc}}$ is a depth-$d$, size-$s$, everywhere $\tau$-influential tree that satisfies $\dist(f,T^{\star}_{\mathrm{trunc}}) \le \opt_s + O(\eps)$. Therefore, by~\Cref{claim: BuildDT optimal} $\BuildDT$ returns a tree $T$ that also satisfies $\dist(f,T)\le \opt_s + O(\eps)$.

\paragraph{Runtime.} We have the following analogue of~\Cref{claim: runtime}:  
\begin{claim}[Runtime in the monotone case]
\label{lem: runtime monotone}
Assume that variable influences of $f$ and its subfunctions can be computed exactly in unit time. For all $d,s \in \N$ and $\tau > 0$, the algorithm $\BuildDT_M(f, \varnothing, s, d, \tau)$ runs in time  $n \cdot s^2 \cdot  (1/\tau)^{O(d)}$. 
\end{claim}

\begin{proof}
By \Cref{cor: monotone influential vars}, we have that for any restriction $f_\pi$, the set $S$ defined on Step 4(a) of the algorithm has size at most $|S| \le \lfrac{1}{4\tau^2}$. The rest of the proof proceeds exactly as in \Cref{claim: runtime}, where $\lfrac{\log s}{\tau}$ is replaced by $\lfrac{1}{4\tau^2}$. This gives a bound of 
\[s \cdot \left(\frac{1}{\tau}\right)^{O(d)}\]
for the number of recursive calls, and a bound of 
\[n \cdot s^2 \cdot \left(\frac{1}{\tau}\right)^{O(d)}\]
for the total running time.
\end{proof}

Finally, we remove the assumption that the variable influences of $f$ and its restrictions can be computed in unit time.  We claim that they can be efficiently estimated to sufficiently high accuracy using only uniform random labeled examples $(\bx,f(\bx))$. As in the proof of~\Cref{lem: main realizable}, it suffices to ensure that all the estimates that our algorithm makes are accurate to within $\pm\frac{\tau}{2}$.  

Using \Cref{fact: influence correlation}, we have for any $i$ and restriction $\pi$,
\[\Inf_i(f_\pi) = \lfrac{1}{2}\E[f_\pi(\bx)\bx_i] = \lfrac{1}{2}\E[f(\bx)\bx_i ~|~ \bx \text{ consistent with } \pi].\]
The right hand side is equivalent to 
\[\frac{\lfrac{1}{2}\E \big[f(\bx)\bx_i \cdot \mathbbm{1}[\bx \text{ consistent with } \pi]\big]}{\Pr[\bx \text{ consistent with } \pi]} = 2^{|\pi| - 1} \cdot \E \big[f(\bx)\bx_i \cdot \mathbbm{1}[\bx \text{ consistent with } \pi]\big].\]
To estimate $\Inf_i(f_\pi)$ to accuracy $\pm \lfrac{\tau}{2}$, it then suffices to estimate $\E\big[ f(\bx)\bx_i \cdot \mathbbm{1}[\bx \text{ consistent with } \pi] \big]$ to accuracy $\pm \tau \cdot 2^{-|\pi|}$. By Chernoff bounds, this can be estimated with confidence $\ge 1 - \delta$ with 
\[O\left(\frac{1}{\tau^2} \cdot 2^{2d} \cdot \log (1 / \delta)\right)\]
uniform random examples $(\bx,f(\bx))$ labeled by $f$, where we have used the fact that $|\pi| \le d$. Each estimate takes time 
\[n \cdot O\left(\frac{1}{\tau^2} \cdot 2^{2d} \cdot \log (1 / \delta)\right).\]
The number of times variable influences are computed during the execution of $\BuildDT$ is at most $n \cdot (\lfrac{1}{\tau})^{O(d)}$, so by setting  $\delta < 1/(n \cdot (\lfrac{1}{\tau})^{O(d)})$ we ensure that w.h.p.~all our estimates are indeed accurate to within $\pm \frac{\tau}{2}$.  The sample complexity of our algorithm is 
\[O\left(\frac{1}{\tau^2} \cdot 2^{2d} \cdot \big(\log n + d\log\lfrac{1}{\tau}\big)\right) =  \poly(s/\eps)\cdot \log n,\]
and by~\Cref{lem: runtime monotone}, the overall runtime of our algorithm is 
\[ n\cdot s^2 \cdot \left(\frac1{\tau}\right)^{O(d)} \cdot \frac{n\cdot 2^{2d}}{\tau^2} \cdot \left(\log n + d\log(1/\tau)\right) \le  \tilde{O}(n^2)\cdot (s/\eps)^{O(\log((\log s)/\eps))}. \] 
This completes the proof of~\Cref{thm:monotone}.


\label{sec:monotone}

%% file: proper.bbl
\newcommand{\etalchar}[1]{$^{#1}$}
\begin{thebibliography}{DCRT19}

\bibitem[AH12]{AH12}
Micah Adler and Brent Heeringa.
\newblock Approximating optimal binary decision trees.
\newblock {\em Algorithmica}, 62(3-4):1112--1121, 2012.

\bibitem[BFJ{\etalchar{+}}94]{BFJKMR94}
Avirm Blum, Merrick Furst, Jeffrey Jackson, Michael Kearns, Yishay Mansour, and
  Steven Rudich.
\newblock Weakly learning {DNF} and characterizing statistical query learning
  using {F}ourier analysis.
\newblock In {\em Proceedings of the 26th Annual ACM Symposium on Theory of
  Computing (STOC)}, pages 253--262, 1994.

\bibitem[BGLT20]{BGLT-NeurIPS1}
Guy Blanc, Neha Gupta, Jane Lange, and Li-Yang Tan.
\newblock Universal guarantees for decision tree induction via a higher-order
  splitting criterion.
\newblock In {\em Proceedings of the 34th Conference on Neural Information
  Processing Systems (NeurIPS)}, 2020.

\bibitem[BL97]{BL97}
Avrim Blum and Pat Langley.
\newblock Selection of relevant features and examples in machine learning.
\newblock {\em Artificial Intelligence}, 97(1-2):245--271, 1997.

\bibitem[BLT20]{BLT-ITCS}
Guy Blanc, Jane Lange, and Li-Yang Tan.
\newblock Top-down induction of decision trees: rigorous guarantees and
  inherent limitations.
\newblock In {\em Proceedings of the 11th Innovations in Theoretical Computer
  Science Conference (ITCS)}, volume 151, pages 1--44, 2020.

\bibitem[Blu92]{Blu92}
Avrim Blum.
\newblock Rank-{$r$} decision trees are a subclass of {$r$}-decision lists.
\newblock {\em Inform. Process. Lett.}, 42(4):183--185, 1992.

\bibitem[Bsh93]{Bsh93}
Nader Bshouty.
\newblock Exact learning via the monotone theory.
\newblock In {\em Proceedings of 34th Annual Symposium on Foundations of
  Computer Science (FOCS)}, pages 302--311, 1993.

\bibitem[CM19]{CM19}
Sitan Chen and Ankur Moitra.
\newblock Beyond the low-degree algorithm: mixtures of subcubes and their
  applications.
\newblock In {\em Proceedings of the 51st Annual ACM Symposium on Theory of
  Computing (STOC)}, pages 869--880, 2019.

\bibitem[DCRT19]{DRT19}
Hugo Duminil-Copin, Aran Raoufi, and Vincent Tassion.
\newblock Sharp phase transition for the random-cluster and potts models via
  decision trees.
\newblock {\em Annals of Mathematics}, 189(1):75--99, 2019.

\bibitem[EH89]{EH89}
Andrzej Ehrenfeucht and David Haussler.
\newblock Learning decision trees from random examples.
\newblock {\em Information and Computation}, 82(3):231--246, 1989.

\bibitem[GKK08]{GKK08}
Parikshit Gopalan, Adam Kalai, and Adam Klivans.
\newblock Agnostically learning decision trees.
\newblock In {\em Proceedings of the 40th ACM Symposium on Theory of Computing
  (STOC)}, pages 527--536, 2008.

\bibitem[Han93]{Han93}
Thomas Hancock.
\newblock {Learning $k$$\mu$ decision trees on the uniform distribution}.
\newblock In {\em Proceedings of the 6th Annual Conference on Computational
  Learning Theory (COT)}, pages 352--360, 1993.

\bibitem[HJLT96]{HJLT96}
Thomas Hancock, Tao Jiang, Ming Li, and John Tromp.
\newblock Lower bounds on learning decision lists and trees.
\newblock {\em Information and Computation}, 126(2):114--122, 1996.

\bibitem[HKY18]{HKY18}
Elad Hazan, Adam Klivans, and Yang Yuan.
\newblock Hyperparameter optimization: A spectral approach.
\newblock {\em Proceedings of the 6th International Conference on Learning
  Representations (ICLR)}, 2018.

\bibitem[JS06]{JS06}
Jeffrey~C. Jackson and Rocco~A. Servedio.
\newblock On learning random dnf formulas under the uniform distribution.
\newblock {\em Theory of Computing}, 2(8):147--172, 2006.

\bibitem[JZ11]{JZ11}
Rahul Jain and Shengyu Zhang.
\newblock The influence lower bound via query elimination.
\newblock {\em Theory of Computing}, 7(1):147--153, 2011.

\bibitem[KM93]{KM93}
Eyal Kushilevitz and Yishay Mansour.
\newblock Learning decision trees using the {Fourier} spectrum.
\newblock {\em SIAM Journal on Computing}, 22(6):1331--1348, December 1993.

\bibitem[KS06]{KS06}
Adam Klivans and Rocco Servedio.
\newblock Toward attribute efficient learning of decision lists and parities.
\newblock {\em Journal of Machine Learning Research}, 7(Apr):587--602, 2006.

\bibitem[KST09]{KST09}
Adam Kalai, Alex Samorodnitsky, and Shang-Hua Teng.
\newblock Learning and smoothed analysis.
\newblock In {\em Proceedings of the 50th Annual IEEE Symposium on Foundations
  of Computer Science (FOCS)}, pages 395--404, 2009.

\bibitem[Lee10]{Lee10}
Homin~K. Lee.
\newblock Decision trees and influence: an inductive proof of the osss
  inequality.
\newblock {\em Theory of Computing}, 6(4):81--84, 2010.

\bibitem[LR76]{LR76}
Hyafil Laurent and Ronald Rivest.
\newblock Constructing optimal binary decision trees is {NP}-complete.
\newblock {\em {Information Processing Letters}}, 5(1):15--17, 1976.

\bibitem[MOS04]{MOS04}
Elchanan Mossel, Ryan O'Donnell, and Rocco~A. Servedio.
\newblock Learning functions of {$k$} relevant variables.
\newblock {\em Journal of Computer and System Sciences}, 69(3):421--434, 2004.

\bibitem[MR02]{MR02}
Dinesh Mehta and Vijay Raghavan.
\newblock Decision tree approximations of boolean functions.
\newblock {\em Theoretical Computer Science}, 270(1-2):609--623, 2002.

\bibitem[O'D14]{ODBook}
Ryan O'Donnell.
\newblock {\em Analysis of Boolean Functions}.
\newblock Cambridge University Press, 2014.

\bibitem[OS07]{OS07}
Ryan O'Donnell and Rocco Servedio.
\newblock {Learning monotone decision trees in polynomial time}.
\newblock {\em SIAM Journal on Computing}, 37(3):827--844, 2007.

\bibitem[OSSS05]{OSSS05}
Ryan O'Donnell, Michael Saks, Oded Schramm, and Rocco Servedio.
\newblock Every decision tree has an influential variable.
\newblock In {\em Proceedings of the 46th Annual IEEE Symposium on Foundations
  of Computer Science (FOCS)}, pages 31--39, 2005.

\bibitem[Riv87]{Riv87}
Ronald Rivest.
\newblock Learning decision lists.
\newblock {\em Machine learning}, 2(3):229--246, 1987.

\bibitem[Sie08]{Sie08}
Detlef Sieling.
\newblock Minimization of decision trees is hard to approximate.
\newblock {\em Journal of Computer and System Sciences}, 74(3):394--403, 2008.

\bibitem[Val15]{Val15}
Gregory Valiant.
\newblock Finding correlations in subquadratic time, with applications to
  learning parities and the closest pair problem.
\newblock {\em Journal of the ACM (JACM)}, 62(2):1--45, 2015.

\bibitem[ZB00]{ZB00}
Hans Zantema and Hans Bodlaender.
\newblock Finding small equivalent decision trees is hard.
\newblock {\em International Journal of Foundations of Computer Science},
  11(2):343--354, 2000.

\end{thebibliography}
